\documentclass[acmsmall,nonacm]{acmart}
\renewcommand\footnotetextcopyrightpermission[1]{}

\AtBeginDocument{%
  }

\usepackage{amsmath}
\usepackage[ruled,vlined,linesnumbered]{algorithm2e}

\usepackage{graphicx}
\usepackage{textcomp}
\usepackage{xcolor}
\usepackage{xspace}
\usepackage{pifont}
\usepackage{makecell}
\usepackage{booktabs}
\usepackage{multirow}
\usepackage{array}
\usepackage{subcaption}
\usepackage{enumitem}
\usepackage{float}
\usepackage{colortbl}
\usepackage{diagbox}
\usepackage{listings}
\usepackage[most]{tcolorbox}
\newcommand{\techname}{TestDecision\xspace}

\usepackage{amsthm}
\newtheorem{definition}{Definition}

\newtheorem{theorem}{Theorem}
\newtheorem{remark}{Remark}
\newtheorem{assumption}{Assumption}

\begin{document}

\title{TestDecision: Sequential Test Suite Generation via Greedy Optimization and Reinforcement Learning}

\author{Guoqing Wang}
\email{guoqingwang@stu.pku.edu.cn}
\affiliation{%
  \institution{Peking University}
  \city{Beijing}
  \country{China}
}

\author{Chengran Yang}
\authornote{Co-corresponding authors.}
\email{cryang@smu.edu.sg}
\affiliation{%
  \institution{Singapore Management University}
  \city{Singapore}
  \country{Singapore}
}

\author{Xiaoxuan Zhou}
\email{zhouxx@mails.neu.edu.cn}
\affiliation{%
  \institution{Northeastern University}
  \city{Shenyang}
  \country{China}
}

\author{Zeyu Sun}
\email{zeyu.zys@gmail.com}
\affiliation{%
  \institution{Institute of Software, Chinese Academy of Sciences}
  \city{Beijing}
  \country{China}
}

\author{Bo Wang}
\email{wangbo_cs@bjtu.edu.cn}
\affiliation{%
  \institution{Beijing Jiaotong University}
  \city{Beijing}
  \country{China}
}

\author{David Lo}
\email{davidlo@smu.edu.sg}
\affiliation{%
  \institution{Singapore Management University}
  \city{Singapore}
  \country{Singapore}
}

\author{Dan Hao}
\authornotemark[1]
\email{haodan@pku.edu.cn}
\affiliation{%
  \institution{Peking University}
  \city{Beijing}
  \country{China}
}

\renewcommand{\shortauthors}{Wang et al.}

\newcommand{\red}[1]{\textcolor{red}{#1}}

\begin{abstract}
With the rapid evolution of Large Language Models (LLMs), automated software testing is witnessing a paradigm shift.
While proprietary models like GPT-4o demonstrate impressive capabilities, their high deployment costs and data privacy concerns make open-source LLMs the practical imperative for many academic and industrial scenarios.
In the field of automated test generation, it has evolved to iterative workflows to construct test suites based on LLMs. 
When utilizing open-source LLMs, we empirically observe they lack a suite-level perspective, suffering from \textit{structural myopia}—failing to generate new tests with large marginal gain based on the current covered status.
In this paper, from the perspective of sequences, we formalize test suite generation as a Markov Decision Process (MDP) and demonstrate that its objective exhibits \textbf{monotone submodularity}, which enables an effective relaxation of this NP-hard global optimization into a tractable step-wise greedy procedure.
Guided by this insight, we propose \textbf{\techname}, which transforms LLMs into neural greedy experts.
\techname\ consists of two synergistic components: (1) an inference framework which implements test suite construction following a step-wise greedy strategy; and (2) a training pipeline of reinforcement learning which equips the base LLM with sequential test generation ability to maximize marginal gain.
Comprehensive evaluations on the ULT benchmark demonstrate that \techname\ significantly outperforms existing advanced methods. 
It brings an improvement between 38.15-52.37\% in branch coverage and 298.22-558.88\% in execution pass rate over all base models, achieving a comparable performance on 7B backbone with a much larger proprietary LLM GPT-5.2.
Furthermore, \techname\ can find 58.43-95.45\% more bugs than vanilla base LLMs and exhibit superior generalization on LiveCodeBench, proving its capability to construct high-quality test suites.

\end{abstract}

\maketitle
\vspace{-2mm}
\section{Introduction}
\label{sec:introduction}
Recent advances in Large Language Models (LLMs)~\cite{deepseek-r1, gpt4o} have fundamentally transformed the landscape of automated software testing, sparking numerous works on test case generation~\cite{lemieux2023codamosa, altmayer2025coverup, yang2024advancing,gu2025llm, lee2025learning}. 
The prevailing paradigm in these approaches typically follows an iterative, sequential process: given a program under test, the model generates test cases one by one to progressively construct a test suite~\cite{yuan2024evaluating, ryan2024code}.
These methods represent the early, pioneering attempts at LLM-based test suite generation, demonstrating impressive capabilities in test suite generation.

While existing heuristic approaches~\cite{takerngsaksiri2025pytester, zhang2025automated, altmayer2025coverup} have shown empirical success for this sequential generation process, they lack a formal guarantee regarding how this process quantitatively contributes to the global optimization objective. This theoretical void causes a mathematical ambiguity whether the generated suite is near-optimal or if substantial coverage gaps persist due to sub-optimal sequential choices.
To the best of our knowledge, it is currently unproven whether prompting the model to generate the next test case in sequence, given the status of the current test suite, provides any theoretical guarantees for the optimality of the final test suite.

Viewing this in-sequence generation paradigm through a formal lens, constructing a test suite is not merely a sequence of independent generation tasks, but an MDP where the state represents the current test suite, and the action space consists of generating potential new test cases.
Evidently, a new test case that executes a particular branch is generally more valuable when that branch has not yet been covered by the current suite than exploring those already been covered.
Consequently, the LLMs must make sequential decisions, i.e., iteratively generating test cases with high value (e.g., coverage or bug detection ability on the entire test suite) based on the current checked status of the test suite, which we define as a composite state capturing both the input dimension (i.e., the accumulated coverage of code paths) and the oracle dimension (i.e., the established fault-detection capability via assertions).
Under this MDP formulation, the ultimate goal is to maximize a specific utility function (e.g., code coverage, mutation score) for the entire suite. 
We highlight that this global optimization objective is fundamentally $NP$-hard (proven in Sec.~\ref{sec:problem_solving}).

The $NP$-hard nature of this problem~\cite{church1974maximal} reveals a critical limitation in current state-of-the-art (SOTA) solutions. Most existing LLM-based methods employ unguided exploration even with test suite context, making local choices without a mechanism to approximate the global optimum.

To address this challenge, we present a novel theoretical framework for test suite generation and a step-wise greedy generation pipeline with a theoretical guarantee.
Specifically, we first formally prove that test suite generation illustrates a monotone submodular property~\cite{bordeaux2014tractability}.
This property allows us to invoke the result calculated by Nemhauser et al.~\cite{nemhauser1978analysis}, which states that a greedy algorithm maximizing marginal gain at each step is a mathematically grounded approximation guaranteed to reach at least $(1 - 1/e)$ of the optimal value. 
Concretely, by selecting a test case at each iteration step that yields the largest immediate marginal gain, the well-trained policy is guaranteed to achieve an approximation to the global optimum~\cite{nemhauser1978analysis}.
This relaxation provides the motivation for designing our step-wise greedy generation pipeline with a theoretical guarantee.

Guided by this formal derivation, we then instantiate an inference framework designed to execute this step-wise greedy strategy. 
For each step, we decide to generate a single test case that has the maximum additional checked value (i.e., \emph{marginal gain} detailed in Sec.~\ref{sec:problem_solving}) to the current test suite, treating it as an atomic greedy decision.
We find that applying our inference framework alone can already outperform vanilla and prompt engineering baselines.

However, we observe through a pilot study that without training, performing only inference guidance cannot fully resolve the sequential test generation challenge.
These approaches suffer from \textit{structural myopia} in that LLMs lack the intrinsic sequential generation capability to reliably follow the greedy objective.
Specifically, even with enough feedback information, existing approaches struggle with this task.
This observation indicates that sequential test generation is not a capability that can be reliably elicited through prompting alone; it must be \emph{learned}. 

To bridge this gap, we finally propose a training framework to enable LLMs' sequential test generation ability with step-wise Reinforcement Learning (RL).
We use RL to train an LLM as a decision policy, i.e., in each step, given the current checked status, plus execution feedback and greedy guidance, the trained LLM can generate a single test case that maximizes the marginal gain.
We represent this marginal gain as a composite reward signal that strictly penalizes LLM to generate test cases with invalid execution and rewards for which of higher coverage benefits.

We construct \techname as a synergistic combination of a greedy inference powered by the tuned LLM that augments the ability of sequential test generation.
This synergistic combination enables LLM to aggressively explore unchecked logic and approximate the global optimum while adhering to strict syntax and execution constraints.

We evaluate \techname\ on unleaked real-world and algorithmic test generation benchmarks (i.e., ULT~\cite{huang2025benchmarking} and LiveCodeBench~\cite{jainlivecodebench}), widely used in previous work~\cite{he2025llm, yang2025qwen3, wang2024advanced, liu2024deepseek}.
The results are positive: \techname\ significantly outperforms existing mainstream LLM-based baselines. 
On all evaluated base models, it improves branch coverage by 38.15-65.87\% and execution pass rate by 298.22-558.88\%. The bug detection ability also has an improvement of 58.43-95.45\% on vanilla base LLMs.
Remarkably, our 7B model even achieves a comparable performance with proprietary LLMs GPT-5.2 and demonstrates superior generalization the out-of-distribution cases~\cite{jainlivecodebench}.

In summary, this paper makes the following contributions:
\vspace{-1mm}
\begin{itemize}[leftmargin=*]
    \item \textbf{Theoretical Formalization:} We formalize test suite generation as an MDP and derive a valid relaxation based on monotone submodularity, providing a theoretical guarantee for effective test suite generation.
    \item \textbf{Novel Framework:} We propose \techname, an instantiation of the step-wise greedy generation paradigm, which includes a complementary inference framework and training pipeline ,effectively equipping LLMs with sequential test generation ability.
    \item \textbf{Extensive Evaluation:} We demonstrate that \techname\ achieves state-of-the-art performance among open-source models, effectively generalizing to unseen coding tasks.
\end{itemize}

\section{Motivation and Empirical Analysis}
\label{sec:motivation}

The paradigm of automated test generation has evolved from single-shot synthesis~\cite{yang2024evaluation, schafer2023empirical} to iterative workflows~\cite{yuan2024evaluating}. Recent studies~\cite{ryan2024code, altmayer2025coverup, gu2025llm} have demonstrated that incorporating execution feedback (e.g., coverage reports or test logs) can guide LLMs to enhance performance. 
However, existing successes largely rely on large proprietary foundation LLMs (e.g., GPT-4o~\cite{gpt4o}). While powerful, these models incur high inference costs, introduce data privacy concerns, and remain black-box services that cannot be retrained or hosted locally for domain-specific policy optimization.

\noindent \textbf{The Shift to Open-Source Models.}
There is a growing imperative to democratize software engineering capabilities using open-source models~\cite{hou2024large} that allow for local deployment and specialized training~\cite{shang2025large, yang2024large, yu2024fine}. 
Furthermore, smaller models have advantages in cost, while large proprietary models are black boxes and cannot be trusted in all scenarios.
However, a critical question remains: \textit{Can these smaller, open models effectively leverage iterative feedback?}

In this section, we show that test suite generation is inherently a \textbf{sequential decision process}.
The empirical evidence highlights that current open models lack the sequential decision-making capability to generate valuable test cases iteratively, succumbing to what we term \textit{structural myopia}.

\subsection{Motivation: The Characteristic of Sequential Decision-Making}
\label{subsec:motivation_example}

To understand why test suite generation is a sequential decision problem, consider the real-world function \texttt{correct\_sc\_info} (Figure~\ref{fig:sequential_example}). The function processes cluster information through two independent logic blocks: \textit{Block 1} corrects CPU counts, and \textit{Block 2} computes missing benchmarks.

\begin{figure}[t]
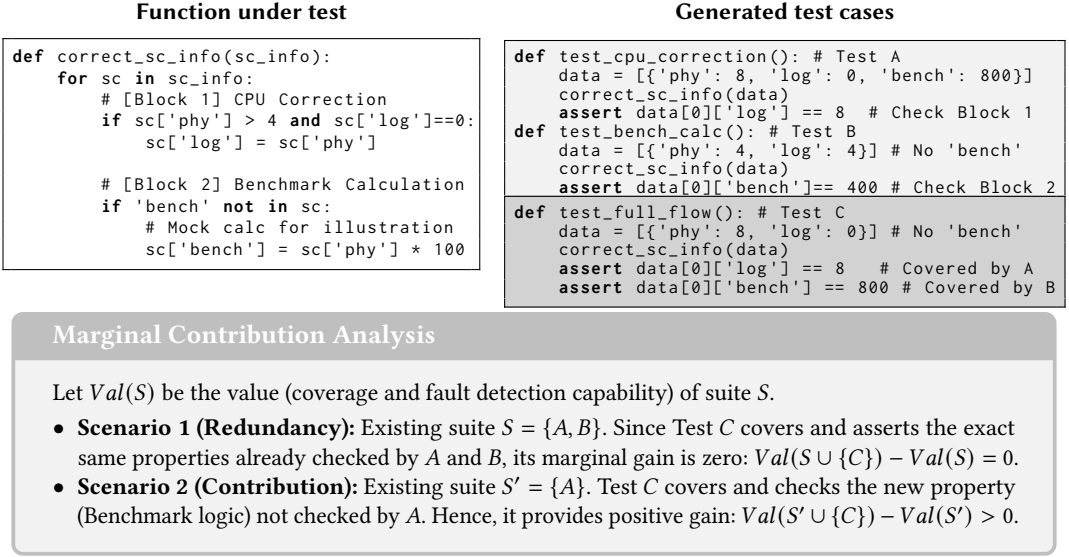

    \centering
        \begin{subfigure}[t]{0.44\textwidth}
     \vspace{0pt}
        \caption*{\textbf{Function under test}}
        \begin{lstlisting}[language=Python, basicstyle=\ttfamily\scriptsize, frame=single]
def correct_sc_info(sc_info):
    for sc in sc_info:
        # [Block 1] CPU Correction
        if sc['phy'] > 4 and sc['log']==0:
            sc['log'] = sc['phy']

        # [Block 2] Benchmark Calculation
        if 'bench' not in sc:
            # Mock calc for illustration
            sc['bench'] = sc['phy'] * 100
        \end{lstlisting}
    \end{subfigure}
    \hfill
    \begin{subfigure}[t]{0.52\textwidth}
    \vspace{0pt}
    \caption*{\textbf{Generated test cases}}
    \begin{lstlisting}[
        language=Python, 
        basicstyle=\ttfamily\scriptsize, 
        frame=tlr, 
        backgroundcolor=\color{gray!10}, 
        belowskip=0pt,
        lineskip=-1pt
    ]
def test_cpu_correction(): # Test A
    data = [{'phy': 8, 'log': 0, 'bench': 800}]
    correct_sc_info(data)
    assert data[0]['log'] == 8  # Check Block 1
def test_bench_calc(): # Test B
    data = [{'phy': 4, 'log': 4}] # No 'bench'
    correct_sc_info(data)
    assert data[0]['bench']== 400 # Check Block 2
    \end{lstlisting}
    \begin{lstlisting}[
        language=Python, 
        basicstyle=\ttfamily\scriptsize, 
        frame=single, 
        backgroundcolor=\color{gray!35}, 
        aboveskip=0pt,
        lineskip=-1pt
    ]
def test_full_flow(): # Test C
    data = [{'phy': 8, 'log': 0}] # No 'bench'
    correct_sc_info(data)
    assert data[0]['log'] == 8   # Covered by A
    assert data[0]['bench'] == 800 # Covered by B
    \end{lstlisting}
\end{subfigure}
    
    \vspace{-3mm}
    \begin{tcolorbox}[colback=gray!10, colframe=gray!50, title=\textbf{Marginal Contribution Analysis}]
    \small
    Let $Val(S)$ be the value (coverage and fault detection capability) of suite $S$.
    \begin{itemize}[leftmargin=*]
        \item \textbf{Scenario 1 (Redundancy):} Existing suite $S = \{A, B\}$. 
        Since Test $C$ covers and asserts the exact same properties already checked by $A$ and $B$, its marginal gain is zero: $Val(S \cup \{C\}) - Val(S) = 0$.
        \item \textbf{Scenario 2 (Contribution):} Existing suite $S' = \{A\}$. 
        Test $C$ covers and checks the new property (Benchmark logic) not checked by $A$. Hence, it provides positive gain: $Val(S' \cup \{C\}) - Val(S') > 0$.
    \end{itemize}
    \end{tcolorbox}

    \caption{Illustrative example of the sequential dependency. Test C is functionally valid, but its contribution depends entirely on whether Tests A and B already exist in the suite.}
    \label{fig:sequential_example}

\end{figure}

This example demonstrates that the utility (i.e., the marginal contribution to checking completeness) of a generated test is not intrinsic; it depends on the current state of the test suite. To generate effective test cases, an ideal system must generate a test case go beyond what has already been checked (Scenario 2) while avoiding redundancy (Scenario 1).
Consequently, this state-dependency proves that test suite generation is fundamentally a sequential decision process. Given an evolving test suite, the decision of what test case to generate next cannot be made in isolation but must account for the current checked status. 
Next, we will investigate whether current models exhibit this sequential test suite generation capability.

\subsection{Pilot Study}
\label{subsec:pilot_setup}

To rigorously assess the capabilities of current LLMs, we design a pilot study comparing different generation paradigms. We utilize three representative open-source models: Qwen2.5-Coder-7B~\cite{hui2024qwen2}, CodeLlama-7B~\cite{roziere2023code}, and Llama-3.1-8B~\cite{dubey2024llama}. 
To avoid potential data leakage~\cite{jainlivecodebench}, we evaluate on unleaked test case generation benchmark, ULT~\cite{huang2025benchmarking}, which consists of 3,909 function-level unit test generation tasks without corresponding test cases 
through a rigorous filtering process on real-world Python functions sourced from The Stack v2~\cite{lozhkov2024starcoder}.

We formally define three generation paradigms to simulate varying levels of generation capability:
(1) \textbf{Direct Generation (Non-Iterative):} The model is prompted to generate a complete test suite including $k$ test cases in a single model call. This serves as a baseline for the model's raw generation capacity without iteration.
(2) \textbf{Iterative (Blind):} The test suite is constructed iteratively. At step $t$, the model generates the $t$-th test case given the context of previous test cases $\{t_1, \dots, t_{t-1}\}$. Without providing any feedback information, this setting tests the model's ability to diversify tests based solely on iteration strategy and conversation history.
(3) \textbf{Iterative (Guided):} This setting mimics current feedback-driven workflows on proprietary models~\cite{gu2025llm, altmayer2025coverup}. At step $t$, the feedback information that guides LLMs to generate new tests to cover uncovered lines and branches is appended to the prompt for step $t+1$, explicitly instructing the model to target these missing parts.

We follow the standard setting, metrics, and prompts used in ULT and other benchmarks~\cite{huang2025benchmarking,wang2025testeval, jaintestgeneval}, where $k$ is set to 5. For evaluation metrics, we employ \textbf{line coverage} and \textbf{branch coverage} to measure the thoroughness of the generated suite, and utilize the \textbf{mutation score} to measure the fault-detection capability. Additionally, we report \textbf{syntax correctness} and \textbf{execution correctness} (the percentage of generated tests that compile and execute successfully) to monitor the hallucinations or invalid code structures.

\vspace{0.5em}
\noindent \textit{Observation: Structural Myopia.}
In Figure~\ref{fig:pilot_study}, we first observe a characteristic logistic-like pattern across all models and settings. In the initial phase ($k=1$ to $2$), the models rapidly cover the structural backbone of the code (e.g., happy paths), leading to a steep rise in coverage. However, as the suite size increases ($k \ge 3$), the marginal coverage gain diminishes rapidly. The curves flatten into a plateau, indicating that simply allocating more generation budget does not enable the models to penetrate complex, nested branches. 

\begin{figure}[t]
    \centering
    \includegraphics[width=\textwidth]{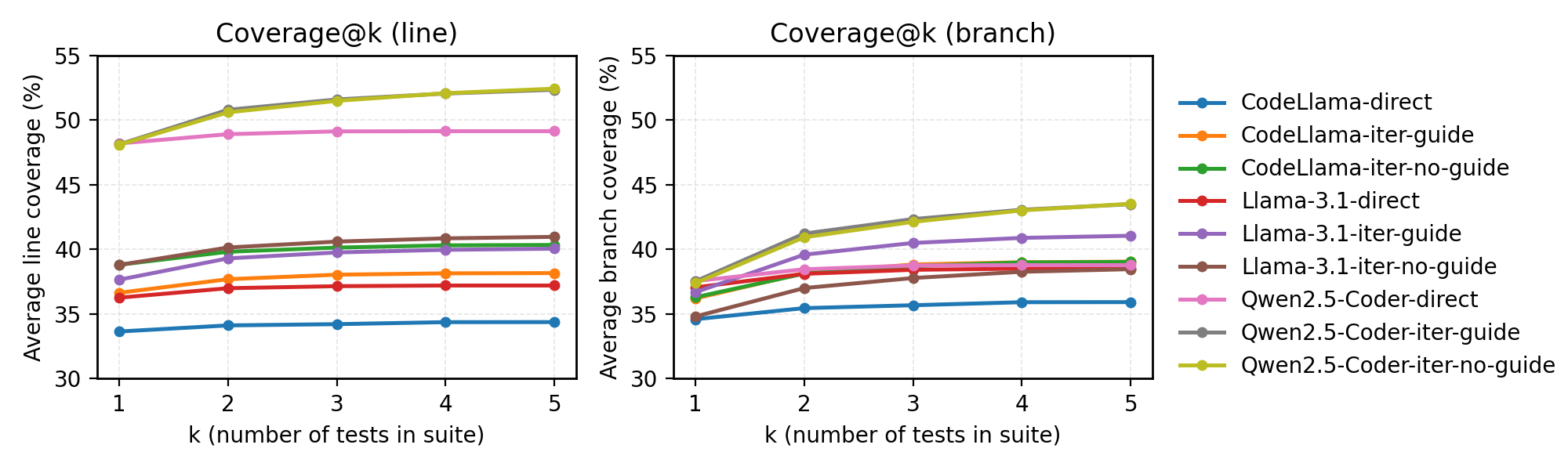} 
    \caption{Coverage growth with respect to test suite size ($k$). All models exhibit a rapid plateau, characteristic of diminishing returns. Notably, the iter-guided lines do not distinctly separate from the iter-blind lines.}
    \label{fig:pilot_study}
\end{figure}

Another critical finding is the inefficacy of explicit feedback information. As shown in Table~\ref{tab:pilot_results}, providing the model with coverage feedback (\textit{Iter-Guided}) fails to yield consistent or substantial improvements over blind iteration (\textit{Iter-Blind}), and in some cases even slightly degrades performance.

For the strongest evaluated model, Qwen2.5-Coder-7B, the guided setting achieves 52.34\% line coverage, which is numerically indistinguishable from, and even slightly lower than the blind setting (52.43\%). The mutation score also reflects this stagnation in bug detection utility, remaining unchanged at exactly 33.81\% in both settings. 
For CodeLlama-7B, line coverage even drops from 40.34\% (Blind) to 38.17\% (Guided), indicating that additional feedback even destabilizes generation for weaker models.
While Llama-3.1-8B shows a marginal increase in branch coverage, this comes at a steep cost to correctness. The syntax correctness plummets from 96.52\% to 90.51\%, and execution correctness remains low. This suggests that when forced to cover specific gaps via prompting, the model tends to hallucinate invalid syntax or unexecutable logic to satisfy the textual constraint.

We summarize these observations as \textit{Structural Myopia}. Without a look-ahead sequential decision policy, the generation process of LLMs quickly exhausts easy targets. Additionally, for base models, the prompt feedback is not effectively exploited, but may act more as noise than as a strategic guide, leaving the underlying greedy trajectory unchanged.

\begin{table*}[t]
\centering
\caption{Pilot study results of Direct Generation, Blind Iteration, and Guided Iteration.}
\label{tab:pilot_results}
\scriptsize
\begin{tabular}{ll|cc|c|cc}
\toprule
\textbf{Model} & \textbf{Setting} & \textbf{Line Cov.} & \textbf{Branch Cov.} & \textbf{Mutation Score} & \textbf{Syntax Corr.} & \textbf{Exec Corr.} \\
\midrule
\multirow{3}{*}{Qwen2.5-Coder-7B} 
  & Direct       & 49.13\% & 38.77\% & 32.61\% & \textbf{99.35\%} & \textbf{18.18\%} \\
  & Iter-Blind   & \textbf{52.43\%} & \textbf{43.52\%} & \textbf{33.81\%} & 99.13\% & 16.89\% \\
  & Iter-Guided  & 52.34\% & 43.49\% & \textbf{33.81\%} & 98.85\% & 16.43\% \\
\midrule
\multirow{3}{*}{CodeLlama-7B}
  & Direct       & 34.38\% & 35.92\% & 26.42\% & \textbf{93.83\%} & \textbf{8.94\%} \\
  & Iter-Blind   & \textbf{40.34\%} & 39.03\% & \textbf{28.07\%} & 89.08\% & 8.66\% \\
  & Iter-Guided  & 38.17\% & \textbf{39.05\%} & 27.68\% & 87.31\% & 8.95\% \\
\midrule
\multirow{3}{*}{Llama-3.1-8B}
  & Direct       & 37.20\% & 38.51\% & 26.79\% & \textbf{96.90\%} & 9.39\% \\
  & Iter-Blind   & \textbf{40.97\%} & 38.47\% & \textbf{28.14\%} & 96.52\% & \textbf{13.98\%} \\
  & Iter-Guided  & 40.05\% & \textbf{41.05\%} & 27.06\% & 90.51\% & 13.48\% \\
\bottomrule
\end{tabular}
\end{table*}

\section{Theoretical Framework and Relaxation}
\label{sec:problem_solving}

In Section~\ref{subsec:motivation_example}, we identify the sequential nature of test suite generation. To provide a rigorous foundation for our approach, we formally model this task as a global optimization problem over MDP~\cite{puterman1990markov}.
We then demonstrate that despite the computational intractability of the global objective, the underlying problem structure allows for a theoretically grounded relaxation. By proving the \textit{monotone submodularity} of the objective function, we derive a step-wise greedy strategy that provides a provable approximation guarantee.

\noindent \textbf{Problem Formalization.}
We first define the test generation process as a finite-horizon MDP.

\begin{definition}[Test Generation MDP]
The test suite generation process is defined by the tuple $\mathcal{M} = (\mathcal{S}, \mathcal{A}, \mathcal{T}, \mathcal{R}, \gamma)$, where:

 \textbf{State Space ($\mathcal{S}$):} A state $s_t \in \mathcal{S}$ encapsulates the static program semantics $\mathcal{C}$ and the dynamic \textbf{checked status} of the accumulated suite $S_{t-1}$. Formally, $s_t$ represents the set of logical units (e.g., lines, branches, or latent bugs) verified so far.
 
\textbf{Action Space ($\mathcal{A}$):} 
    An action $a_t \in \mathcal{A}$ represents the generation of a new, valid test case from the universe of all possible valid tests $\mathcal{U}$, which is treated as a single atomic action.
    
 \textbf{Transition ($\mathcal{T}$):} The transition is deterministic. Given state $s_t$ (implying suite $S_{t-1}$) and action $a_t$, the new suite becomes $S_t = S_{t-1} \cup \{a_t\}$, and the state updates to $s_{t+1}$ reflecting the new checked status.

\end{definition}

\noindent To ensure the theoretical validity of our subsequent derivation, we impose the following necessary assumptions on the execution environment and utility function.

\begin{assumption}[Determinism and Independence]
\label{ass:independence}
Let $F: 2^{\mathcal{U}} \to \mathbb{R}_{\geq 0}$ be the utility function (e.g., coverage or mutation score, or any other measure of check completeness). We assume:
(1) \textbf{Determinism:} The execution result and utility of any test case $a \in \mathcal{U}$ are deterministic and time-invariant;
(2) \textbf{Permutation Invariance:} The utility depends solely on the set of unique test cases, regardless of generation order. Formally, for any sequence $\mathbf{a} = [a_1, \dots, a_K]$, $F(\mathbf{a}) \equiv F(\{a_1, \dots, a_K\})$.
\end{assumption}

\noindent Based on these definitions, the goal of the agent is to find a policy $\pi$ generating a sequence of $K$ actions to maximize the final set utility:
\begin{equation}
    \label{eq:global_obj}
    \pi^* = \mathop{\arg\max}_{\pi} F(S_K), \quad \text{where } |S_K| \le K
\end{equation}

\noindent \textbf{Theoretical Analysis and Relaxation.}
Directly solving Eq.~\ref{eq:global_obj} is computationally intractable.
\begin{remark}[Hardness]
The space of valid test cases $\mathcal{U}$ is effectively infinite. Even if we restrict to a large but finite set $\mathcal{U}_{finite} \subset \mathcal{U}$, the problem of selecting $K$ tests to maximize the utility function $F(S)$ is equivalent to the Maximum Coverage Problem~\cite{khuller1999budgeted}, which is known to be \textbf{NP-hard}~\cite{church1974maximal}. 
\end{remark}

\noindent Leveraging the property of monotone submodularity, we can formally relax the intractable global optimization into a greedy strategy with a provable performance bound:

\begin{theorem}[Approximation Guarantee]
\label{thm:guarantee}
Let $\pi_{greedy}$ be a policy that selects the action maximizing the immediate marginal gain at each step $t$:
$
    a_t = \mathop{\arg\max}_{a \in \mathcal{U}} \left( F(S_{t-1} \cup \{a\}) - F(S_{t-1}) \right)
$

Under Assumption~\ref{ass:independence}, the test suite $S_K$ generated by $\pi_{greedy}$ guarantees a utility of at least $(1 - 1/e) \approx 63.2\%$ of the optimal global policy $\pi^*$.
\end{theorem}

\begin{proof}
The proof proceeds in two steps: first establishing that the objective function exhibits monotone submodularity, and then deriving the approximation bound based on this property.

\noindent \textbf{Step 1: Monotone Submodularity.}
We first show that the utility function $F(S)$ is monotonically non-decreasing and submodular.
Let $E$ be the universe of all structural elements (e.g., lines or branches) to be checked. For any test case $t$, let $Check(t) \subseteq E$ be the set of elements covered by $t$. The utility is defined as $F(S) = |\bigcup_{t \in S} Check(t)|$.
Consider two test suites $A \subseteq B \subseteq \mathcal{U}$ and a new test $x \in \mathcal{U} \setminus B$.
The marginal gain of adding $x$ to a suite $S$ corresponds to the number of elements in $Check(x)$ that remain uncovered in $S$.
Since $A \subseteq B$, the set of covered elements in $A$ is a subset of those in $B$ (i.e., $\bigcup_{t \in A} Check(t) \subseteq \bigcup_{t \in B} Check(t)$). Consequently, any element covered by $x$ that is \textit{new} to $B$ must also be \textit{new} to $A$.
Mathematically, the set difference satisfies:
\begin{equation}
    Check(x) \setminus \bigcup_{t \in B} Check(t) \subseteq Check(x) \setminus \bigcup_{t \in A} Check(t)
\end{equation}
Taking the cardinality of both sides, we obtain the submodularity condition:
\begin{equation}
    F(B \cup \{x\}) - F(B) \leq F(A \cup \{x\}) - F(A)
\end{equation}
Additionally, since adding a test case never removes existing coverage, $F(S)$ is trivially monotone.

\noindent \textbf{Step 2: Derivation of the Bound.}
The sequential generation problem is equivalent to maximizing the set function $F(S)$ under a cardinality constraint $|S| \le K$.
Since we have established in Step 1 that $F(S)$ is a monotone submodular function, we can invoke the classic result by Nemhauser et al.~\cite{nemhauser1978analysis}. Their theorem states that for such functions, the greedy algorithm—which iteratively selects the element with the largest marginal gain—achieves an approximation ratio of $(1 - 1/e)$.
Thus, the step-wise greedy policy $\pi_{greedy}$ provides this theoretical lower bound for the global objective defined in Eq.~\ref{eq:global_obj}.
\end{proof}

\noindent This theoretical guarantee allows us to relax the intractable sequential optimization into a \textbf{step-wise greedy optimization} problem to make this problem solvable. 
When the utility function is monotone submodular, a step-wise greedy strategy by selecting the item with the largest marginal gain in each step is sufficient to approximate the global optimum efficiently. In the following section, we instantiate test suite generation as such an optimization problem, and develop an LLM-based inference framework that performs step-wise greedy generation.

\section{Methodology}
\label{sec:methodology}

In this section, we translate this theoretical step-wise greedy strategy into a practical LLM-based test generator \techname. 
The design of \techname\ consists of two components:
\begin{itemize}[leftmargin=*]
    \item \textbf{An inference procedure} that constructs test suite following a step-wise greedy strategy.
    \item \textbf{A training framework} that equips the base LLM with \emph{sequential} test generation ability.
\end{itemize}

We first instantiate the step-wise greedy objective by formalizing the notions of step, testing state, action, and marginal gain for test suite generation.
Based on this instantiation, we present the inference procedure of \techname.
This procedure implements the step-wise greedy strategy, which is theoretically justified under monotone submodularity and provides an effective approximation to the global optimum.
Finally, we introduce our training framework, which trains the LLM with step-level reinforcement learning so that it can consistently generate non-redundant tests with higher marginal gains under evolving testing states.

Overall, our inference and training are complementary.
The inference procedure follows the step-wise greedy strategy and provides an effective approximation to the global optimum, while the RL training strengthens the base LLM so that it can more reliably execute this strategy in practice under diverse testing states.

\subsection{Instantiation}
We instantiate test generation as a step-wise greedy decision process by mapping it to the MDP
formulation introduced in Section~\ref{sec:problem_solving}.
Concretely, we define:
\begin{itemize}[leftmargin=*]
    \item \textbf{Step.} A step corresponds to generating a single test case. Each test-case generation is treated as an atomic decision.
    \item \textbf{Testing state.} It is the checked status of the program induced by the current test suite.
    The checked status is derived by executing the program with the accumulated tests and reflects which parts of the program have already been exercised and which remain unchecked at this step.
    \item \textbf{Marginal gain.} The marginal gain of an action is the additional checked value it brings when added to the current test suite.
    We define the detailed checked value in Section~\ref{sec: margin}.
\end{itemize}

This instantiation turns test suite generation into a sequential decision-making problem.
At each step, the LLM observes the current testing state, generates a candidate test case, and receives execution feedback, which updates the testing state for the next step.
Given a budget of $K$ tests, the objective is to iteratively generate the test that maximizes marginal gain, thereby approximating the global optimum under monotone submodularity.

Notably, this instantiation preserves the monotonicity and submodular structure assumed in our theoretical analysis.
First, it is \emph{monotone}: adding a valid test case never reduces the checked value of the suite, e.g., coverage and mutation scores do not decrease, and previously checked behaviors remain valid.
Second, it exhibits \emph{submodularity}: it is easier to obtain high-gain tests when the suite is small, but increasingly harder as the suite grows and common program behaviors become covered.
These properties ensure that our instantiation conforms to the assumptions in Section~\ref{sec:problem_solving}, making step-wise greedy generation a valid and effective approximation to the global objective.

\begin{figure}[t]
    \centering
    \includegraphics[width=0.75\textwidth]{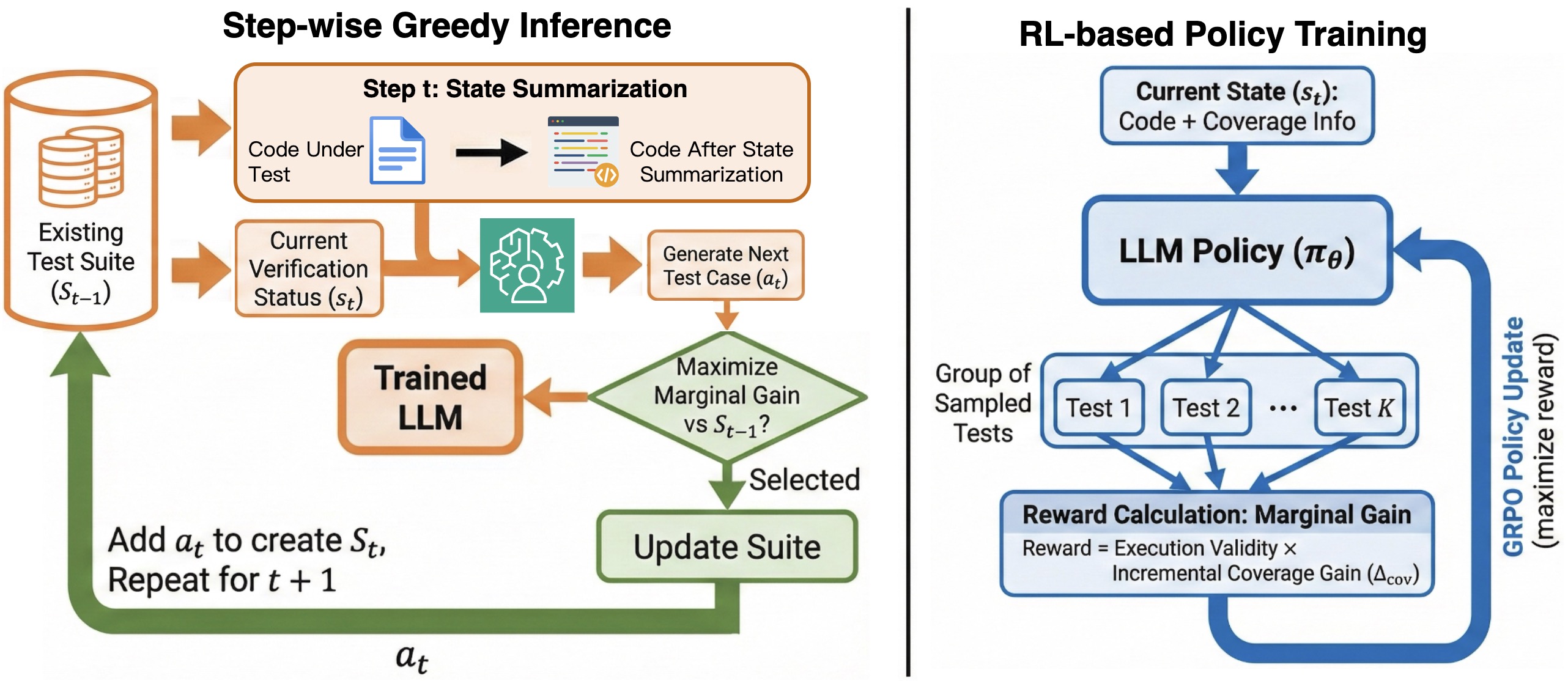}
    \caption{The overview of \techname. The framework operates as an iterative generation loop. 
    The LLM acts as a greedy policy generator, observing the code state augmented with checked markers. The environment executes the generated test, checks its correctness, and updates the state for the next step, ensuring the LLM always targets the remaining checked frontier.
    }
    \label{fig:framework}
    \vspace{-2mm}
\end{figure}

\subsection{Inference Procedure}
\label{subsec:overview}

Given the above instantiation, we design an inference procedure for \techname\ that follows the step-wise greedy objective as illustrated in Fig.~\ref{fig:framework}.
At inference time, \techname\ starts from an empty test suite $S_0=\emptyset$ and iteratively constructs a $K$-test suite under a budget of $K$ actions.

At each step $t$, \techname\ executes the current suite $S_{t-1}$ to obtain execution feedback and summarizes it into the current testing state $s_t$.
Conditioned on $s_t$ and the program under test, the LLM proposes the next action by generating a candidate test case $a_t$ with the goal of maximizing the marginal gain.
\techname\ then adds $a_t$ to the suite, forming $S_t = S_{t-1} \cup \{a_t\}$, and executes the updated suite to obtain the next state $s_{t+1}$.
This process repeats until $K$ test cases are generated.

\subsubsection{State Summarization (Perception).} 
At each step $t$, \techname\ observes the current testing state $s_t$, which summarizes the checked status of the program induced by the current test suite.
To make an informed greedy decision, the LLM must understand what has already been executed and, more importantly, what remains unchecked.

However, raw execution feedback (e.g., coverage reports and logs) is often too verbose and redundant to be directly used in prompting, increasing the cognitive workload of the base model.
We therefore convert $s_t$ into a compact, code-grounded state summary via \emph{in-context coverage projection}: we annotate the code with syntactic markers that highlight the currently uncovered regions (with minimal surrounding context) without changing program semantics and execution status, thereby exposing the checked frontier that the next test should target.

Specifically, we built a standard dynamic instrumentation tool wrapping \texttt{Coverage.py}~\cite{coveragepy_docs}. For any state $s_t$, the environment executes the current suite $S_{t-1}$ to identify the set of uncovered line numbers $U_t$. It then performs Abstract Syntax Tree (AST) injection, parsing the code and injecting the marker \texttt{"\#uncovered"} as a trailing comment to every node in $U_t$. Finally, we convert the modified AST back to source code. This ensures the signal is structurally valid and aligned with the code logic when fed into the LLM.

\subsubsection{Greedy Action} 
With the state summary in place, \techname\ performs greedy action selection by using the LLM as a stochastic policy.
Formally, at step $t$, the LLM induces a policy $\pi_\theta(a \mid s_t)$ over the test-case space, and an action corresponds to generating one complete test case $a_t$.
The intent of this action is greedy: conditioned on the uncovered markers exposed in $s_t$, the model is prompted to focus its generation on maximizing the coverage of currently unchecked regions and on asserting the corresponding behaviors. 
Once $a_t$ is generated, \techname\ executes it together with $S_{t-1}$ to update the suite and the testing state for the next step.

\subsection{Policy Training}
\label{subsec:policy_optimization}

Our empirical study in Section~\ref{subsec:pilot_setup} reveals a fundamental gap between the sequential nature of test suite construction and the poor sequential generation ability of existing open LLMs.
Even when equipped with explicit execution feedback, existing LLMs fail to escape from the diminishing-returns plateau, repeatedly generating structurally redundant tests and exhibiting structural myopia.
This observation indicates that sequential test generation is not a capability that can be reliably elicited through \emph{prompting} alone; it must be \emph{learned}.

\subsubsection{Training Objective.}
Crucially, the goal is not to teach the LLM to write high-quality test cases in isolation.
Instead, we seek to learn a \emph{decision policy} that determines \emph{which} test to generate under a given testing state.
As illustrated in Section~2, the value of a test is not intrinsic but conditional on the existing suite: a test becomes less useful once its behavior is already covered. Following the relaxation in Section~3, the intractable global objective of maximizing the final suite utility can be approximated by a step-wise greedy strategy that maximizes the marginal gain at each step.
This transforms long-horizon planning into a sequence of local decisions and suggests that the policy should be optimized directly against the \emph{immediate marginal gain} of each generated test.

\subsubsection{RL vs. SFT for Sequential Test Generation.}
Under this formulation, the training objective is to maximize the marginal gain of each generated test given the current suite.
This requirement fundamentally mismatches the learning signal provided by commonly-used supervised fine-tuning (SFT)~\cite{ouyang2022training, tufano2020unit}.
SFT trains the model to reproduce a single target test under a given prompt, which provides no supervision on such marginal gain.
When multiple tests are all valid, it offers no signal indicating which one contributes more new coverage or avoids redundancy.
As a result, SFT alone cannot reliably teach an LLM to perform sequential decision-making in test generation.

To provide the missing learning signal, we turn to reinforcement learning.
Instead of asking the model to imitate a fixed test case, at each RL training step, we let the target LLM generate a test case under the current testing state, execute it, and measure its marginal gain with respect to the existing suite.
This marginal gain is formulated as a discrete reward and is used to update the model’s policy:
generating test cases with higher marginal gain are reinforced, while redundant or ineffective ones are penalized.
In this way, the LLM is trained to optimize its policy toward generating test cases that
consistently maximize marginal gain.

\subsubsection{Design RL Reward.}
\label{sec: margin}
The key to applying RL in our setting is to design a reliable reward $r(s_t, a_t)$ that
approximates the marginal gain of each generated test case given
the current suite:
\begin{equation}
    \label{eq:single_step_reward}
    r(s_t, a_t) = F(S_{t-1} \cup \{a_t\}) - F(S_{t-1})
\end{equation}

By grounding the learning signal in this marginal contribution, the reward provides a concise and faithful guide for the LLM, steering it toward generating test cases that maximize marginal gain under the current state.

\noindent \textbf{Definition of marginal gain}.
We define the marginal gain of a test case in terms of both \emph{program coverage} and \emph{fault-detection} capability, encouraging LLMs to generate test cases that can not only exercise more program logic but also effectively expose semantic faults. 

\noindent \textbf{Gain in program coverage}.
For program coverage, we use line coverage as a proxy.
We quantify how many new code lines are exercised by a new test case.
Concretely, given a program under test with $L_A$ lines, let the existing test suite
$S_{t-1}$ cover $L_S$ lines.
After adding a candidate test $a_t$, we execute the updated suite and observe that it now covers
$L_N$ lines in total ($L_N \geq L_S$).
We then define the marginal coverage gain of $a_t$ as
$
\Delta_{\text{cov}}(a_t) = \frac{L_N - L_S}{L_A - L_S}.
$
This normalization maps the marginal coverage gain to $[0,1]$ and measures how much of the remaining uncovered code
is eliminated by a new test case.
As a result, tests that reach previously unexplored code lines receive a higher reward, whereas tests that only re-execute already covered paths yield zero reward.

\noindent \textbf{Gain in test correctness}.
For the ability of fault detection, mutation testing offers a closer approximation, but it is too expensive to perform online during training. According to our experimental statistics, the average time to obtain mutation score is 76.53 seconds, which is an unacceptable time overhead.
We therefore adopt execution correctness of the generated assertion as a lightweight proxy.
This proxy is justified because execution validity is a strict prerequisite for fault detection—a test case cannot reveal semantic logic errors if it fails to execute or crashes prematurely. 
By prioritizing executable assertions, we ensure the model establishes the necessary behavioral oracle required to capture potential faults.
Specifically, we treat execution correctness as a global control signal in the reward.
A test case that fails to compile or execute receives a zero reward.
Only when a test executes successfully do we further evaluate its contribution in terms of coverage. 

\noindent \textbf{Reward definition}.
The final reward $r(s_t, a_t)$ combines execution correctness and marginal coverage gain:
\begin{equation}
\label{eq:rl_reward}
    r(s_t, a_t) = \underbrace{\Delta_{\text{cov}}(a_t)}_{\text{Coverage Gain}} \cdot \underbrace{\mathbb{I}(\text{Valid}(a_t))}_{\text{Assertion Validity}}
\end{equation}

This design treats execution correctness as a hard prerequisite.
Invalid tests receive a zero reward regardless of their nominal coverage, while only
successfully executed tests are further rewarded by how much new code they cover.
As a result, the policy discourages generating unexecutable or fragile tests and is
explicitly guided to prefer generating test cases that are both valid and yield higher marginal coverage gains.

\subsubsection{RL training}
After finalizing the reward design, we adopt Group Relative Policy Optimization (GRPO)~\cite{shao2024deepseekmath} to train the target LLM. 
At each training step, given the same testing state $s_t$ (including the program under test labeled with coverage signals as discussed in Section~\ref{sec:problem_solving}), the policy samples a group of
candidate test cases $\{a_t^{(1)}, \ldots, a_t^{(G)}\}$.
Each candidate is executed to obtain its reward $r(s_t, a_t^{(i)})$ as defined above.
GRPO updates the policy by comparing candidates within the same group.
Specifically, we compute the relative advantage of each test case as: 
\begin{equation}
A_t^{(i)} = r_t^{(i)} - \frac{1}{G} \sum_{j=1}^{G} r_t^{(j)},
\end{equation}
\noindent which measures how much better a candidate performs than the group average under the same state.
The action to generate test cases with positive advantage is reinforced, while those with negative advantage are punished.
The GRPO objective is then defined as
\begin{equation}
\mathcal{L}_{\text{GRPO}}(\theta) =
- \mathbb{E}_{s_t, \{a_t^{(i)}\}}
\left[
\frac{1}{G} \sum_{i=1}^{G}
A_t^{(i)} \cdot \log \pi_\theta(a_t^{(i)} \mid s_t)
\right]
+ \beta \, \mathrm{KL}\!\left(\pi_\theta \,\|\, \pi_{\text{base}}\right),
\end{equation}
where $\pi_{\text{base}}$ denotes the frozen base model and $\beta$ controls the strength of the
KL regularization.
This formulation trains the policy to increase the probability of test cases that achieve higher marginal gains relative to other candidates generated under the same state, while preventing excessive drift from the base model.

\subsection{Training Data Construction}
\label{subsec:data_construction}

To enable effective policy optimization, our RL model requires access to a diverse set of valid and reachable intermediate states to explore from.
Hence, we construct a series of trajectories of valid states and extract step-wise data from them as training data. Unlike SFT, we use these trajectories solely to initialize valid intermediate states for the RL agent to explore from.

\noindent \textbf{1. Candidate Sampling.}
We first collect a pool of candidate test cases, which can originate from either existing oracles or any external generator. Notably, our framework is agnostic to the source of candidate tests. For the ULT benchmark, Huang et al.~\cite{huang2025benchmarking} do not release oracle tests due to data leakage concerns, so we instantiate this step using an LLM-based generator. 
In our training phase, we randomly select 70\% of ULT tasks for trajectory construction.
To balance cost and performance, we employ \textit{GPT-5-mini} to generate test cases for each given function. For each task, we follow the prompting pipeline provided by Huang et al. to sample $N = 20$ independent test cases. Note that GPT-5-mini is only used once to populate this candidate pool. Our RL training never relies on its logits or preferences, and any sufficiently diverse generator could be substituted here.

\noindent \textbf{2. Greedy Ordering.}
To mimic an optimal planning process, we apply the greedy ordering described in Section~\ref{sec:problem_solving}.
We iteratively select the test case that provides the maximum marginal gain relative to the currently selected set. We then perform state replay as follows. We execute the sorted test sequence step-by-step against the code. This allows us to capture the exact coverage at each step. This ensures that every state $s_t$ represents a reachable and real program state, avoiding the impossible states that might arise from random masking or static estimation.
To mitigate the effects of noisy data, we only keep the code and the corresponding test suites if the final test suite achieves $>90\%$ line coverage. 
Finally, we decompose each checked trajectory of length $L$ into $L$ distinct states $\{s_0, \dots, s_{L-1}\}$. After this processing, our training data includes a total of 8,396 pieces of data.

\section{Experimental Setup}
\label{sec:experimental_setup}

\subsection{Research Questions}
We evaluate \techname by answering the following research questions:

\noindent \textbf{RQ1 (Effectiveness):} Does \techname generate better test suites than base models under advanced prompting with feedback-driven refinement, and other learning-based strategies?

\noindent \textbf{RQ2 (Optimization in Sequence):} Does \techname achieve a higher coverage/mutation score with fewer test cases (a steeper growth trajectory) and test the backbone logic first?

\noindent \textbf{RQ3 (Generalization and Bug Detection):} Can the policy learned by \techname generalize to unseen scenarios, and does improved coverage translate into tangible fault-detection capability? We evaluate performance on LiveCodeBench~\cite{jainlivecodebench} to examine robustness on unseen coding tasks.

\noindent \textbf{RQ4 (Ablation Study):} What is the contribution of each component in \techname? 

\subsection{Evaluation Benchmarks}
Considering the well-known impact of data leakage, we employ two distinct benchmarks to evaluate in-distribution performance and out-of-distribution generalization.

\noindent \textbf{ULT Benchmark.} 
As described in Section~\ref{subsec:data_construction}, we randomly split the ULT tasks into 70\% for training data construction and hold out the remaining 30\% solely for evaluation. The evaluation reflects the model's ability to handle unseen code within a similar domain distribution.

\noindent \textbf{LiveCodeBench (LCB).}
To prevent memorization, we utilize LCB, focusing on tasks released after June 2024, which are constructed after the reported training cutoff of all selected base models, serving as a strong proxy for intrinsic reasoning capability. For each task, we treat the canonical reference solution as the program under test when generating tests, and then evaluate bug detection by executing the generated suites against both the canonical and the associated buggy solutions.

\subsection{Baselines}
\label{subsec:baselines}

To validate the effectiveness of \techname, we compare it against the following three groups of baselines: foundation models, advanced prompting strategies, and learning-based approaches.

\noindent \textbf{Group I: Zero-Shot Open-Source LLMs.}
We first evaluate widely used open-source models (i.e., \texttt{Qwen2.5-Coder-7B}~\cite{hui2024qwen2}, \texttt{Llama-3.1-8B}~\cite{dubey2024llama}, and \texttt{CodeLlama-7B}~\cite{roziere2023code}) without any additional training to establish a performance floor. For fairness, we reuse the official generation pipeline provided by the ULT benchmark for all three models, which employs an iterative guided prompting strategy and a fixed test budget $K=5$ utilized in previous work~\cite{wang2025testeval,huang2025benchmarking}.

Our pilot study and the results in previous work~\cite{huang2025benchmarking} show that Qwen2.5-Coder-7B outperforms the other two base models. Therefore, in subsequent experiments involving more sophisticated prompting frameworks, we instantiate them with Qwen2.5-Coder-7B as the underlying backbone to ensure a strong and consistent foundation.

\noindent \textbf{Group II: Advanced Prompting Strategies.}
We further compare against advanced prompting frameworks. These methods represent strong baselines under frozen-weights deployment. We select approaches that are open-source and suitable for Python, rather than being tightly coupled to other languages ~\cite{altmayer2025coverup, yang2024advancing, gu2025llm}, to avoid artifacts introduced by cross-language migration~\cite{baltajicross, ahmed2022multilingual}.

\texttt{ChatTester~\cite{yuan2024evaluating}:} A ChatGPT-style agent approach that improves test generation through a dedicated iterative refinement mechanism. It incorporates a checker to analyze compilation and execution errors, feeding this information back to the LLM to repair the test cases dynamically.
 
\texttt{SymPrompt~\cite{ryan2024code}:} A path-aware prompting strategy that leverages execution analysis to identify uncovered paths. It enriches the prompt with specific path constraints, guiding the LLM to generate inputs that exercise hard-to-reach branches.

\noindent \textbf{Group III: Learning-based Approach.}
We compare \techname against \texttt{TestCTRL}~\cite{zhang2025automated}, which employs proximal policy optimization (PPO)~\cite{schulman2017proximal} with a learned value network (Critic) and utilizes Chain-of-Thought (CoT)~\cite{wei2022chain} traces to guide exploration. We follow the authors' recommended hyperparameters for PPO and CoT generation. To ensure a fair comparison, we train it on the same backbone (Qwen2.5-Coder-7B).

\subsection{Evaluation Metrics}
We employ five complementary metrics, widely used in prior work, to assess the quality and effectiveness of generated test suites. (1) \textbf{Syntactic Correctness Rate:} This metric measures the fraction of generated test cases that are syntactically valid Python code. We implement this check using Python's built-in parser \texttt{ast.parse}. (2) \textbf{Execution Pass Rate:} This metric measures the fraction of generated test \emph{cases} that execute without runtime errors under \texttt{pytest}. (3) \textbf{Line and Branch Coverage:}
For each task, we aggregate the coverage of all tests in the suite and report the average across the benchmark. To analyze RQ2, we additionally compute the coverage trajectory as the suite size $k$ increases from 1 to 5 (Coverage@$k$).
(4) \textbf{Mutation Score:} High structural coverage does not necessarily imply strong fault-detection capability~\cite{inozemtseva2014coverage}. Following existing work~\cite{mundler2024swt, huang2025benchmarking}, we therefore measure the percentage of artificial mutants killed by the generated test suites using Cosmic-Ray~\cite{cosmicray_docs}. 
(5) \textbf{Bug Detection Rate:} Besides mutation score, we further consider semantic bugs that are more representative of real developer mistakes. 
On LCB, we evaluate the real-world bug detection rate by executing each generated suite against both the canonical reference solution and its associated buggy solutions to answer RQ3. A bug is considered detected if there exists at least one test that passes on the canonical solution but fails on the buggy solution.

\subsection{Implementation Details}
\label{subsec:implementation_details}

We train \textbf{\techname} using Qwen2.5-Coder-7B, Llama-3.1-8B, and CodeLlama-7B as base models. 
All models are trained on 8 $\times$ NVIDIA A100 80GB GPUs utilizing the Verl framework~\cite{sheng2024hybridflow}.
We follow the instructions of Verl and perform a small grid search to identify stable hyperparameters. Specifically, the final configuration is detailed as follows.
We set the group sampling size $G=8$.  We use the AdamW optimizer with a learning rate of $1\mathrm{e}{-6}$ and a cosine decay schedule. To prevent the policy from deviating too far from the base model's linguistic distribution, we set the KL divergence coefficient to $\beta=0.001$. 
To prevent infinite loops or malicious code execution during online training, execution is confined within a sandboxed container with strict resource limits (3s timeout).

\section{Experimental Results}
\label{sec:results}

In this section, we present the experimental results to answer the RQs formulated in Section~\ref{sec:experimental_setup}. 

\subsection{RQ1: Effectiveness of \techname}
\label{subsec:rq1_results}

We first evaluate the overall effectiveness of \techname.  To ensure robustness, all experiments are repeated three times with different seeds, and we report the averaged results. Table~\ref{tab:main_results} summarizes the performance on the ULT benchmark. We also conduct a Wilcoxon signed-rank test~\cite{wilcoxon1945individual} to validate statistical significance ($p < 0.01$) for all improvements reported below.
 
\begin{table*}[t]
\centering
\caption{Main Results on ULT Benchmark. \techname significantly outperforms all open-source baselines, including advanced prompting strategies (ChatTester, SymPrompt) and learning-based methods (TestCTRL).}

\label{tab:main_results}
\scriptsize
\begin{tabular}{l|ccc|cc}
\toprule
\textbf{Approach} & \textbf{Line Cov. (\%)} & \textbf{Branch Cov. (\%)} & \textbf{Mutation Score (\%)} & \textbf{Syntactic (\%)} & \textbf{Execution (\%)} \\
\midrule
CodeLlama-7B & 38.17 & 39.05 & 27.72 & 87.31 & 8.95 \\
Llama-3.1-8B & 40.05 & 41.05 & 27.04 & 90.51 & 13.48 \\
Qwen2.5-Coder-7B & 51.42 & 43.23 & 33.80 & 98.66 & 16.86 \\
\midrule
SymPrompt & 34.50 & 28.01 & 16.42 & 99.59 & 25.02 \\
ChatTester &  44.15 & 34.64  & 22.79  & 99.18  & 6.14\\
\midrule
TestCTRL & 43.14 & 55.95  &  35.27 & 74.86  & 12.02 \\
\midrule
\textbf{\techname-CodeLlama} & \textbf{63.48} & \textbf{57.94} & \textbf{47.02} & \textbf{99.87} & \textbf{58.97} \\
\textbf{\techname-Llama} & \textbf{62.72} & \textbf{56.71} & \textbf{45.96} & \textbf{100.00} & \textbf{58.26} \\
\textbf{\techname-Qwen} & \textbf{69.69} & \textbf{65.87} & \textbf{57.13} & \textbf{99.95} & \textbf{67.14} \\
\bottomrule
\end{tabular}

\end{table*}

\techname consistently and substantially outperforms the corresponding base models across all metrics.
Taking the strongest backbone Qwen2.5-Coder-7B as a primary example, \techname lifts line coverage from 51.42\% to 69.69\% and branch coverage from 43.23\% to 65.87\%. This corresponds to a remarkable relative improvement of 35.53\% and 52.37\%, respectively.
The execution pass rate nearly quadruples, soaring from 16.86\% to 67.14\% (a \textbf{298.22\% relative improvement}), while syntactic correctness is maintained at a near-perfect 99.95\%.
This indicates that our framework effectively suppresses the hallucination of unexecutable code, a common pitfall in base models.

This trend generalizes robustly to other backbones.
For CodeLlama-7B and Llama-3.1-8B, \techname achieves a 66.31\% and 56.60\% relative gain in line coverage, respectively.
Notably, both models see an exponential improvement in execution correctness (from $\sim$9-13\% to $\sim$58\%), demonstrating that \techname successfully aligns even weaker models with the strict constraints of valid test generation.
We further examine whether higher coverage translates to better bug-finding capability.
\techname achieves the highest mutation scores among all open-source methods (e.g., \techname-Qwen reaches 57.13\%). This confirms that the greedy generation expert we trained is genuinely exploring deeper logic paths where subtle bugs reside.

Group~II results reveal the limitations of iterable prompt engineering alone.
ChatTester and SymPrompt fail to surpass \techname, even underperforming the base Qwen2.5 backbone in coverage (e.g., SymPrompt achieves only 34.50\% line coverage).
This aligns with our pilot study (Section~\ref{subsec:pilot_setup}): static prompting often forces models to hallucinate invalid logic to satisfy constraints, leading to a collapse in execution rate (e.g., ChatTester drops to 6.14\%). 
In contrast, \techname optimizes the generation policy itself, enabling it to navigate complex constraints intrinsically rather than superficially.
Under the same Qwen2.5-Coder backbone, \techname surpasses TestCTRL by significant margins: approximately +61.54\% relative improvement in line coverage and +17.73\% in branch coverage.
This result empirically validates our theoretical relaxation in Section~\ref{sec:problem_solving}: directly optimizing the step-wise greedy objective is more effective than training a separate value network without execution in the sequential generation task.

\subsection{RQ2: Optimization in Sequence}
\label{subsec:rq2_results}

\begin{figure}[t]
    \centering
    
    \includegraphics[width=0.75\textwidth]{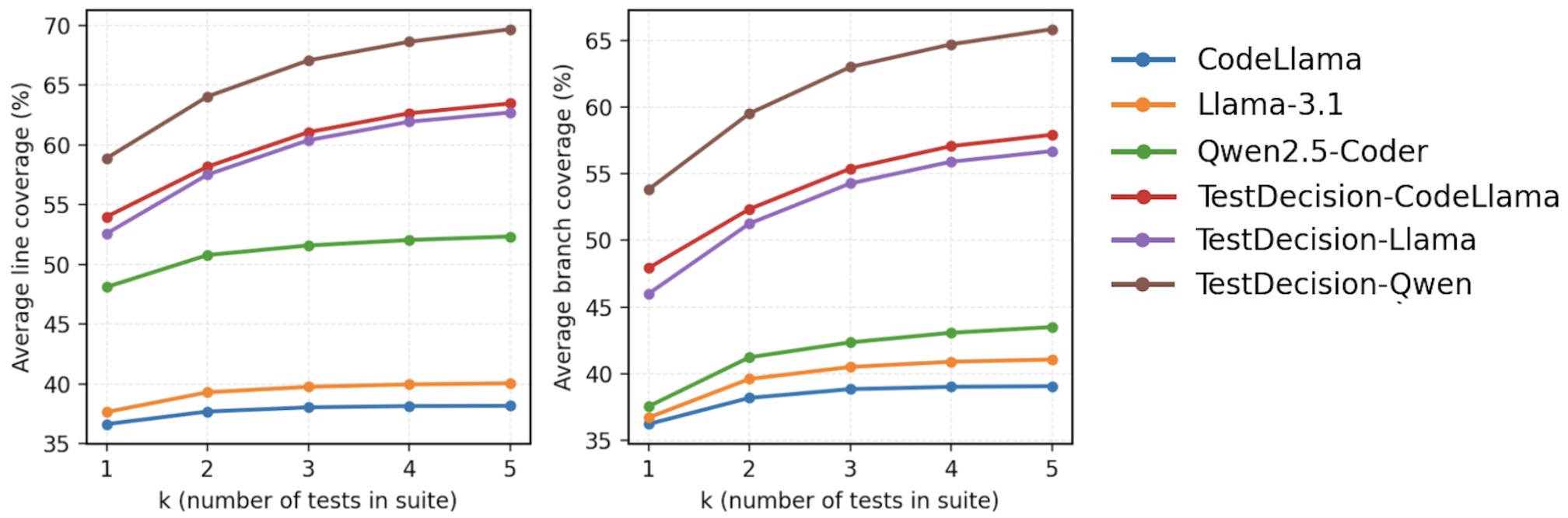} 

    \caption{Step-wise performance trajectory. \techname exhibits a steeper growth curve compared to baselines. 
    }
    \label{fig:RQ_2_coverage_at_k}
     \vspace{-2mm}
\end{figure}

In this RQ, we further investigate whether \techname can mitigate the structural myopia of base models.
We analyze the step-wise generation trajectory on ULT. Figure~\ref{fig:RQ_2_coverage_at_k} visualizes the growth of line and branch coverage as the test suite size $k$ increases from 1 to 5.

As evident in Figure~\ref{fig:RQ_2_coverage_at_k}, \techname-Qwen (Red Line) exhibits the steepest growth trajectory among all methods.
Most notably, at the very first step ($k=1$), our model already achieves a line coverage of 58.27\%, which is significantly higher than the base Qwen2.5-Coder (47.31\%).
This indicates that \techname has learned to prioritize: it does not randomly generate a test; instead, it identifies the backbone logic of the given focal method and generates the most high-yield test case first. This capability effectively reduces the inference budget required to reach a target coverage level. \techname surpasses the final performance of the base model ($k=5$) using only 1 test case.

A critical limitation of existing LLMs identified in Sec.~\ref{subsec:pilot_setup} is the rapid saturation of marginal gain. We observe that baselines flatten out after $k=3$, indicating a failure to find the remaining corner cases.
In contrast, \techname maintains a positive slope even at $k=4$ and $k=5$.
Quantitatively, from $k=3$ to $k=5$, the base Qwen2.5-Coder model only gains 0.3\% coverage, implying redundancy. In the same interval, \techname gains 2.6\% (nearly $9\times$ higher marginal gain).
This confirms that the policy has learned to leverage the checked frontier markers to target the remaining unchecked logic purposefully, rather than hallucinating redundant tests for already-covered paths.

\subsection{RQ3: Generalization and Bug Detection}
\label{subsec:rq3_results}

To answer RQ3, we assess whether the policy learned by \techname\ can generalize to unseen scenarios and effectively detect bugs. We evaluate performance on LCB~\cite{jainlivecodebench}, serving as a rigorous test for out-of-distribution (OOD) generalization, ensuring no data leakage.

\begin{table*}[t]
\centering
\caption{Generalization Results on LiveCodeBench. 
}

\label{tab:lcb_results}
\scriptsize
\resizebox{0.97\textwidth}{!}{%
\begin{tabular}{l|ccc|cc|c}
\toprule
\textbf{Approach} & \textbf{Line Cov. (\%)} & \textbf{Branch Cov. (\%)} & \textbf{Mutation Score (\%)} & \textbf{Syntactic (\%)} & \textbf{Execution (\%)} & \textbf{Bug Det. (\%)}\\
\midrule
CodeLlama-7B & 12.24 & 31.44 & 16.29 & 64.20 & 5.78 & 8.50\\
Llama-3.1-8B & 27.61 & 44.00 & 22.17 & 86.94 & 23.90 & 11.00\\
Qwen2.5-Coder-7B & 49.61 & 50.52 & 38.34 & 93.88 & 36.36 & 20.50\\
\midrule
SymPrompt (Qwen) & 19.09 & 13.33 & 11.28 & 97.39 & 19.69 & 4.50\\
ChatTester (Qwen) & 36.38 & 27.99 & 20.51 & 99.79 & 8.68 & 9.50\\
\midrule
TestCTRL (Qwen) & 44.12 & 53.85  & 36.49  &  78.63 & 8.72 & 24.00\\
\midrule
\textbf{\techname-CodeLlama} & \textbf{42.78} & \textbf{38.03} & \textbf{27.72} & \textbf{99.88} & \textbf{34.80} & \textbf{16.00}\\
\textbf{\techname-Llama} & \textbf{43.62} & \textbf{52.90} & \textbf{34.26} & \textbf{100.00} & \textbf{39.80} & \textbf{21.50}\\
\textbf{\techname-Qwen} & \textbf{62.47} & \textbf{63.51} & \textbf{54.35} & \textbf{100.00} & \textbf{57.73} & \textbf{32.50}\\
\bottomrule
\end{tabular}
}
\end{table*}

As shown in Table~\ref{tab:lcb_results}, \techname\ exhibits strong generalization capabilities, consistently outperforming foundation models across all metrics.
For the Qwen2.5-Coder-7B backbone, \techname\ achieves a 25.92\% relative improvement in line coverage and 25.71\% relative improvement in branch coverage.
This confirms that the sequential generation ability learned via \techname\ captures the intrinsic logic of test suite generation rather than merely overfitting to the training distribution of ULT.
The gains are even more pronounced for weaker backbones.
For CodeLlama-7B and Llama-3.1-8B, we observe a 249.51\% and 57.99\% relative improvement in line coverage.

Additionally, LCB tasks are algorithmically complex, often causing standard models to generate syntactically correct but functionally broken code (e.g., runtime errors).
CodeLlama has an execution pass rate of only 5.78\%, and even Qwen2.5 achieves only 36.36\%.
\techname-Qwen achieves an execution rate of 57.73\% (a 58.77\% relative improvement), and \techname-CodeLlama improves by 502.08\% (5.78\% $\to$ 34.80\%).
Moreover, our models achieve near-perfect syntactic correctness ($\ge 99.88\%$), eliminating the compilation failures that plague baselines like CodeLlama.
On this benchmark, prompting strategies also perform worse than the base model.
SymPrompt achieves only 19.09\% line coverage, and ChatTester reaches 36.38\%.
This highlights a critical limitation of prompt engineering: complex prompting workflows are often brittle and struggle to generalize to complex algorithms without domain adaptation.
In contrast, \techname\ modifies the underlying policy weights, enabling stable and robust generalization that prompting cannot achieve.

\vspace{0.5em}
\noindent \textbf{Bug Detection Capability.}
As shown in the last column of Table~\ref{tab:lcb_results}, \techname\ consistently demonstrates superior bug detection capabilities compared to all baselines.
\techname-Qwen achieves a bug detection rate of \textbf{32.50\%}, marking a substantial 58.5\% relative improvement over the base Qwen2.5-Coder-7B (20.50\%). Similarly, for the weaker CodeLlama-7B backbone, our framework nearly doubles the detection rate from 8.50\% to 16.00\%. This confirms that the coverage gains translate directly into tangible fault-detection power.
\techname-Qwen also outperforms the PPO-based TestCTRL (24.00\%) by a wide margin (+8.5 points). This suggests that our step-wise greedy objective, which emphasizes immediate execution correctness and marginal gain, is more effective at producing tests that catch subtle bugs than global reward optimization.
In summary, \techname\ can generate smarter tests that are effective at exposing realistic software defects.

\subsection{RQ4: Ablation Study}
\label{subsec:ablation}

To disentangle the contributions of our training strategy and inference framework, we conduct a comprehensive ablation study on the Qwen2.5-Coder-7B backbone. We evaluate three variants: (1) \textbf{w/o RL}, which uses SFT on the same data; (2) \textbf{w/o Greedy Inference}, which removes the greedy strategy prompts and state augmentation during inference for the fully trained model; and (3) \textbf{vanilla LLM}, which applies our inference framework to the frozen base model.

\begin{table*}[t]
\centering
\caption{Ablation Study Results on Qwen2.5-Coder-7B.}

\label{tab:ablation}
\scriptsize
\begin{tabular}{l|ccc|cc}
\toprule
\textbf{Variant} & \textbf{Line Cov. (\%)} & \textbf{Branch Cov. (\%)} & \textbf{Mutation Score (\%)} & \textbf{Syntactic (\%)} & \textbf{Execution (\%)} \\
\midrule
\rowcolor{gray!10} \textbf{Full \techname} & \textbf{69.69} & \textbf{65.87} & \textbf{57.13} & \textbf{99.95} & \textbf{67.14} \\
\midrule
w/o RL & 54.14 & 45.42 & 37.26 & 99.85 & 23.98 \\
w/o Greedy Inference & 63.27 & 57.19 & 51.23 & 99.28 & 62.59 \\
vanilla LLM & 59.29 & 51.72 & 41.17 & 99.97 & 25.05 \\
\bottomrule
\end{tabular}%

\end{table*}

Table~\ref{tab:ablation} reveals a strong synergy between our training and inference components. The full \techname\ achieves the best performance across all metrics.
Comparing the full model with the SFT variant reveals the limitations of supervised learning in this domain.
Although SFT utilizes the exact same high-quality training data as our RL training, it achieves only 54.14\% line coverage and, more critically, a low execution pass rate of 23.98\%.
In contrast, \techname\ achieves 69.69\% line coverage and yields a 180\% relative improvement in execution correctness over SFT.
This stark difference indicates that simply cloning expert traces is insufficient for learning the underlying reasoning logic of test generation. 
When removing the inference framework, the line coverage drops from 69.69\% to 63.27\% and mutation score from 57.13\% to 51.23\%.
This confirms that the greedy inference guide is necessary for this sequential decision-making problem. Even a well-trained policy performs better when the checked frontier is explicitly highlighted in the prompt, reducing the cognitive load of identifying checked gaps.

A significant finding is that applying our inference framework can boost the base model to 59.29\% line coverage, confirming our theoretical analysis described in Sec.~\ref{sec:problem_solving}. This suggests that the structural myopia of base models can be partially mitigated by our inference framework alone. Explicitly marking the checked frontier guides LLMs to explore more effectively, even for a frozen model.
However, without RL training, vanilla LLM still suffers from a low execution rate (25.05\%). It sees the gaps but lacks the coding robustness to check them correctly.

In summary, the full \techname\ combines the best of both: the framework provides the guidance (greedy generation and state augmentation), and the RL training provides the better generation capability and robustness (Execution Correctness), resulting in a complementary performance leap.

\section{Discussion}
In this section, we analyze the performance of \techname\ across different model scales, illustrate the performance of proprietary frontier models and CODAMOSA~\cite{lemieux2023codamosa}. Notably, this comparison between \techname and these is not a fair comparison, because proprietary models have much more parameters and the SBST approach generates a large number of candidate tests, from which iterative evolution and selection are performed.

\begin{table*}[t]
\centering

\caption{Comprehensive Comparison with Proprietary Models, Different Scales, and Traditional Approaches. }

\label{tab:discussion_results}
\scriptsize

\begin{tabular}{l|ccc|cc}
\toprule
\textbf{Approach} & \textbf{Line Cov. (\%)} & \textbf{Branch Cov. (\%)} & \textbf{Mutation Score (\%)} & \textbf{Syntactic (\%)} & \textbf{Execution (\%)} \\
\midrule
GPT-4o & 59.88 & 51.71 & 46.41 & 100.00 & 27.16 \\
GPT-5.2 & 69.40 & 65.65 & 66.02 & 99.97 & 70.58 \\
\midrule
CODAMOSA & 52.73 & 48.60 & 23.56 & 77.88 & 15.26 \\
\midrule
Qwen2.5-3B-Base & 46.58 & 41.69 & 29.35 & 94.27 & 12.75 \\
\textbf{\techname-Qwen-3B} & \textbf{52.37} & \textbf{48.94} & \textbf{37.84} & \textbf{100.00} & \textbf{23.90} \\
\midrule
Qwen2.5-7B-Base & 51.42 & 43.23 & 33.80 & 98.66 & 16.85 \\
\textbf{\techname-Qwen-7B} & \textbf{69.69} & \textbf{65.87} & \textbf{57.13} & \textbf{99.95} & \textbf{67.14} \\
\midrule
Qwen2.5-14B-Base & 55.07 & 46.75 & 37.21 & 99.49 & 18.41 \\
\textbf{\techname-Qwen-14B} & \textbf{74.01} & \textbf{70.14} & \textbf{61.58} & \textbf{99.92} & \textbf{68.82} \\
\bottomrule
\end{tabular}%

\end{table*}

\subsection{Comparison with Proprietary Models}
A key question for the software engineering community is whether open-source models can rival proprietary giants in specialized tasks.
As shown in Table~\ref{tab:discussion_results}, the base Qwen2.5-7B model lags significantly behind GPT-4o (59.88\%) and GPT-5.2 (69.40\%).
However, \textbf{\techname-Qwen2.5-7B} completely reverses this gap, achieving 69.69\% line coverage, which surpasses GPT-4o by 16.38\% and achieves comparable performance with GPT-5.2.
It suggests that for test suite generation, a smaller model (7B) aligned with a rigorous theoretical objective and a suitable training phase can outperform general-purpose frontier models that rely on generic instruction following.

\subsection{Scalability and Scaling Laws}
We investigate whether \techname\ scales effectively with model size. Table~\ref{tab:discussion_results} compares the performance gains across 3B, 7B, and 14B parameter scales.
\techname-Qwen2.5-3B improves line coverage from 46.58\% to 52.37\%. While an improvement, the absolute gain (+5.79\%) is modest compared to larger models. This suggests that very small models may lack the fundamental reasoning capacity to fully exploit the checked frontier markers.
The gain explodes at the 7B level, with line coverage jumping from 51.42\% to 69.69\%. This indicates a phase transition where the model becomes capable enough to strictly follow the greedy policy. When the model scales, the trend continues, with the 14B model achieving the highest performance of 74.01\%.
The results demonstrate that \techname\ is highly scalable. Our framework extracts \textit{larger} gains as the base model becomes more capable, suggesting promising potential for future larger-scale open models.

\subsection{Comparison with Search-Based Software Testing (SBST)}
Finally, we present the performance of CODAMOSA~\cite{lemieux2023codamosa}, a representative method that combines SBST with LLM-based regeneration. 
Since the original CODAMOSA implementation is tailored to Codex~\cite{chen2021evaluating}, we modify its configuration to use the base Qwen2.5-Coder-7B model in our experiments.
CODAMOSA achieves 52.73\% line coverage, which slightly outperforms the base model. 
\techname-Qwen2.5-7B significantly outperforms CODAMOSA by 32.16\%.
More critically, CODAMOSA suffers from a low execution pass rate (15.26\%) and syntactic correctness (77.88\%). Furthermore, even for the test cases that are runnable, the readability is very unacceptable.
In contrast, \techname\ maintains near-perfect syntax (>99.9\%) and high pass rates. This highlights \techname's advantage: it retains the valid code generation capability and readability of LLMs while integrating the systematic coverage optimization, effectively getting the best of both worlds.

\subsection{Threats to Validity}
\label{subsec:threats}
We categorize the potential threats to the validity of our study into three categories and discuss the measures we take to mitigate these threats.

\textbf{Threats to internal validity } mainly lie in the following aspects.
To mitigate the implementation threat, we build upon mature, industry-standard open-source libraries. Furthermore, we conduct rigorous testing and manual checks on the implementation. 
To maximally avoid data leakage issue, we use ULT and LCB in our evaluation.
To ensure fair comparisons, we use the official implementations and recommended settings provided in the original papers for all baselines.

\textbf{Threats to external validity} mainly lie in the following aspects.
When considering generalizability across models and scales, we extend our experiments to cover Llama-3.1, CodeLlama, and Qwen series, as well as models with different parameter sizes ranging from 3B to 14B. 
Currently, our implementation and evaluation focus on Python. This is primarily due to the unleaked benchmark requirement and Python's mature dynamic instrumentation ecosystem. However, our core methodology is theoretically language-agnostic. Future work can transfer this framework to statically typed languages by simply replacing the underlying coverage collection tools.
Although ULT contains real-world open-source functions, it still has limitations compared to system-level testing of industrial-scale software. Complex class dependencies and the need for external environment simulation might pose challenges to the RL environment. Nevertheless, ULT has been demonstrated to be significantly more representative of real-world complexity than traditional toy datasets like HumanEval, covering complex control flows and logic dependencies.

\textbf{Threats to construct validity} mainly lie in the metric reliability.
The existence of equivalent mutants might lead to an underestimation of mutation scores. This is a common and persistent challenge in software testing research. To mitigate this issue, we use a standard subset of mutation operators and adopt relative improvement for comparison. Since all methods are evaluated on the same set of mutants, any bias resulting from equivalent mutants is systematic across all baselines.

\section{Related Work}
Automated unit test generation has long been dominated by Search-Based Software Testing (SBST) techniques~\cite{fraser2011evosuite, fraser2014large, pacheco2007randoop, lukasczyk2022pynguin}. Tools like EvoSuite~\cite{fraser2011evosuite} and Pynguin~\cite{lukasczyk2022pynguin} employ genetic algorithms to evolve test suites. While effective for structural coverage, these methods often produce unreadable code and struggle with complex semantic constraints or hard-to-cover branches guarded by deep logic~\cite{lemieux2023codamosa}. The emergence of LLMs has shifted the paradigm towards generating more human-like and semantically meaningful tests. Early works, such as AthenaTest~\cite{tufano2020unit}, demonstrated the feasibility of fine-tuning models on large corpora of method-test pairs. 
Codamosa~\cite{lemieux2023codamosa} uses LLMs to generate new test inputs to help SBST to escape coverage plateaus.
However, benchmarking studies such as ULT~\cite{huang2025benchmarking} and TestEval~\cite{wang2025testeval} have highlighted significant challenges, including data contamination and the struggle of LLMs to handle structurally complex real-world functions~\cite{mundler2024swt}.

To further address the generation limitations (such as hallucinations and compilation errors)~\cite{yang2024evaluation, schafer2023empirical}, recent research focuses on iterative refinement workflows that incorporate external feedback. ChatTester~\cite{yuan2024evaluating} introduces a multi-turn dialogue framework where ChatGPT refines tests based on static analysis and execution feedback. Similarly, CoverUp~\cite{altmayer2025coverup} and TELPA~\cite{yang2024advancing} employ coverage reports and counterexamples to iteratively guide the LLM toward uncovered code regions. More advanced methods like Panta~\cite{gu2025llm} and SymPrompt~\cite{ryan2024code}, utilize fine-grained program analysis to guide exploration. TestART~\cite{gu2024testart} focuses on automated repair of generated tests to resolve compilation failures. While effective, these inference-time techniques primarily rely on prompt engineering to guide frozen models~\cite{chen2024chatunitest, wang2024hits}, often succumbing to structural myopia.

RL has become a cornerstone for aligning LLMs with rigorous software engineering objectives that are non-differentiable, such as execution correctness and runtime efficiency. In test generation, RLSQM~\cite{steenhoek2025reinforcement}, TestCTRL~\cite{zhang2025automated} employs PPO to conduct RL training, though it relies on a computationally heavy Actor-Critic architecture. 
PyTester~\cite{takerngsaksiri2025pytester} leverages PPO to generate test cases from text descriptions in test-driven development. 
Distinct from these approaches, our work demonstrates that an inference grounded in the submodularity of the coverage objective, and a synergistic training framework based on GRPO, is sufficient to achieve SOTA alignment. 

Beyond testing tasks, recent works like CodeRL~\cite{le2022coderl}, CodeRL+~\cite{jiang2025coderl+}, CUBE~\cite{wang2025co}, RLTF~\cite{liurltf}, and VeriRL~\cite{teng2025verirl} utilize RLVR to bridge the gap between pre-training objectives and execution semantics in code generation. PRLCoder~\cite{ye2025process} further advances this by introducing process-supervised RL to guide the step-by-step reasoning of code synthesis. In the broader scope, RL has also been applied to program fuzzing, where CovRL~\cite{eom2024fuzzing} learn to mutate inputs to maximize feedback signals.

Monotone submodularity has been successfully applied to a wide range of AI and ML problems, such as influence maximization in social networks~\cite{kempe2003maximizing, kempe2005influential}, sensor placement and information-gathering in graphical models~\cite{krause2005near}, data subset selection and core-set construction~\cite{mirzasoleiman2016fast}, and document summarization~\cite{li2012multi} with diversity and coverage constraints.

\section{Conclusion and Future Work}
In this paper, we tackle the coverage plateau in automated test generation by fundamentally reframing the task as a sequential decision-making process.
We demonstrate that the global objective of test suite generation exhibits monotone submodularity. This theoretical insight allows us to relax the intractable sequential generation problem into a feasible step-wise greedy strategy, providing a rigorous mathematical foundation for iterative generation.
Guided by this theory, we present \textbf{\techname}, which bridges the gap between theoretical optimality and practical capability through a complementary inference framework and training pipeline.
Our extensive evaluation on the ULT and LiveCodeBench benchmarks establishes \techname\ as the new state-of-the-art considering base models of the same sizes.
Our 7B-parameter model not only achieves comparable performance with larger proprietary models but also demonstrates robust out-of-distribution generalization. Additionally, as our base model gets larger, \techname outperforms the giant proprietary model at the 14B scale. This highlights that if we increase the size even further, the margin can potentially be larger.
We believe \techname\ sets a new precedent for applying RL to Software Engineering. It indicates that theoretically grounded RL alignment is the key to unlocking the reasoning potential of LLMs for complex, constraint-heavy engineering tasks.

\bibliographystyle{ACM-Reference-Format}
\bibliography{sample-base}

\end{document}